\documentclass[11pt,psamsfonts
]{amsart}
\newcommand{\ignore}[1]{}
\usepackage{a4wide}
\usepackage{lmodern}

\usepackage{amssymb}
\usepackage{amsfonts}
\usepackage{microtype}
\usepackage{amsmath}
\usepackage{booktabs} 
\usepackage{textcomp}
\usepackage{xcolor}
\usepackage[colorlinks]{hyperref}
\definecolor{gruen}{rgb}{0,.5,0}
\hypersetup{colorlinks=true,linkcolor=gruen,citecolor=gruen,urlcolor=gruen}

\usepackage{crossreftools}
\hypersetup{bookmarksnumbered, bookmarksopen=true, bookmarksopenlevel=1}

\usepackage[capitalise,nameinlink]{cleveref}
\crefname{appendix}{Appendix}{Appendices}
\crefname{lem}{Lemma}{Lemmas}
\crefname{thm}{Theorem}{Theorems}
\crefname{rem}{Remark}{Remarks}
\crefname{ex}{Example}{Examples}
\crefname{cor}{Corollary}{Corrolaries}
\crefname{prop}{Proposition}{Propositions}
\crefname{figure}{Figure}{Figures}

\let\origmaketitle\maketitle
\def\maketitle{
  \begingroup
  \def\uppercasenonmath##1{} 
  \let\MakeUppercase\relax 
  \origmaketitle
  \endgroup
}

\newtheorem{thm}{Theorem}
\newtheorem*{nthm}{Eichelberger's Theorem}

\newtheorem{cor}{Corollary}
\newtheorem{lem}{Lemma}
\newtheorem{prop}{Proposition}

\newtheorem{dfn}{Definition}

\theoremstyle{remark}
\newtheorem{rem}{Remark}
\newtheorem{ex}{Example}

\let\bar\overline

\newcommand{\unc}{\mathfrak{u}}
\usepackage{mathabx}
\newcommand{\up}[1]{#1^{\triangledown}}
\newcommand{\Item}[1]{\item[\mbox{\rm (#1)}]} 
\newcommand{\plus}{\mbox{\tiny $+$}}
\newcommand{\pos}[1]{#1_{\plus}} 
\newcommand{\Leq}{\preccurlyeq}

\newcommand{\hL}[1]{L_{\unc}(#1)} 
\newcommand{\mL}[1]{L_{\plus}(#1)} 
\newcommand{\uL}[1]{L(#1)} 
\newcommand{\zL}[1]{L_{\mathrm{z}}(#1)}
\newcommand{\zpL}[1]{L_{\mathrm{pz}}(#1)}
\newcommand{\uper}{\mathrm{per}_n}
\newcommand{\Case}[1]{\smallskip \noindent{\it Case}~#1:}
\newcommand{\tern}{\alpha}
\newcommand{\ttern}{\beta}
\newcommand{\ntern}{\bar{\alpha}}

\newcommand{\dual}[1]{{#1}^d}
\newcommand{\subcube}{A_{\tern}}
\newcommand{\nsubcube}{A_{\bar{\tern}}}
\newcommand{\tsubcube}{A_{\ttern}}

\newcommand{\cl}{c} 
\newcommand{\terncl}{c_{\tern}}
\newcommand{\terntr}{t_{\tern}}
\newcommand{\tterntr}{t_{\ttern}}
\newcommand{\nterntr}{t_{\ntern}}

\newcommand{\U}{H}
\newcommand{\implic}[2]{(#1)~$\Rightarrow$~(#2)}
\newcommand{\equival}[2]{(#1)~$\Leftrightarrow$~(#2)}
\newcommand{\bit}{\epsilon}
\newcommand{\HH}{\mathcal{H}}
\newcommand{\der}[3]{\mathrm{d} #1(#2;#3)}
\newcommand{\xx}{a\oplus\unc x} 

\newcommand{\cliq}{\mathrm{CLIQUE}_{n,k}}
\newcommand{\ncliq}{\mathrm{CLIQUE}}
\newcommand{\ecliq}{f_{n,k}}
\newcommand{\exk}[1]{\mathrm{E}_{#1}}

\usepackage[all,arc,curve,color,frame]{xy}

\def\circuits{
\begin{figure}[t]
\begin{center}
$\xygraph {
    !{(1,-0.5)}*+={\bar{x}_n}="xn"
    !{(1.7,-0.5)}*+={F_0}="F0"
    !{(2.9,-0.5)}*+={F_1}="F1"
    !{(3.5,-0.5)}*+={x_n}="x"
    !{(1,-1)}*+=[o]+[F]{\land}="a"
    !{(2.25,-1)}*+=[o]+[F]{\land}="b"
    !{(3.5,-1)}*+=[o]+[F]{\land}="c"
    !{(1.5,-1.5)}*+=[o]+[F]{\lor}="d"
    !{(2.25,-2)}*+=[o]+[F]{\lor}="e"
    !{(2.25,-2.4)}*+={F}="F"
    "xn":"a" "F0":"a" "F0":"b"
    "F1":"b" "F1":"c" "x":"c" "a":"d" "b":"d" "d":"e" "c":"e"
    }
    \qquad\qquad
\xygraph {
   !{(1.5,-0.2)}*{x_1}="x1"
   !{(2,-0.2)}*{\cdots}="dots"
   !{(2.5,-0.2)}*{x_m}="xm"
    !{(1,-0.4)}*{}="1"
    !{(2,-0.8)}*{H_m}="H"
    !{(3,-0.4)}*{}="2"
    !{(0,-1.5)}*{}="3"
    !{(4,-1.5)}*{}="4"
    !{(2.4,-1.2)}*{2^{2^m}~\mbox{\rm\footnotesize output~functions}}="out"
    !{(0,-2)}*{}="5"
    !{(2.3,-1.8)}*{x_{n-m}}="ym"
    !{(3.3,-1.8)}*{\cdots}="dots2"
    !{(3.8,-1.8)}*{x_{n}}="ym"
    !{(4,-2)}*{}="6"
    !{(2,-2.3)}*{2^{n-m}~\mbox{\rm\footnotesize input~functions}}="in"
    !{(2,-3.5)}*{}="7"
    !{(2,-2.8)}*{F_{n-m}}="Fm"
    !{(1,-1.5)}*+={\bullet}="h"
    !{(1,-1.2)}*+={f_b}="fb"
    !{(0.5,-2)}*+={\bullet}="f1"
    !{(1,-2)}*+={\bullet}="f2"
    !{(1.5,-2)}*+={\bullet}="f3"
    "1"-"2" "1"-"3" "2"-"4" "3"-"4" "5"-"6" "5"-"7" "6"-"7"
    "h"-"f1" "h"-"f2" "h"-"f3"
    }
$
\end{center}
    \caption[]{On the left is a fragment of a circuit for the consensus recursion \cref{eq:rec}, and on the right is a schematic form of a hazard-free circuit which results after combining this recursion with the Shannon-Muller construction. }
    \label{fig:shannon}
\end{figure}
}

\begin{document}

\title[]{\Large Notes on Hazard-Free Circuits}


\author[]{\normalsize Stasys Jukna$^{\ast}$}
\thanks{$^{\ast}$Faculty of Mathematics and Computer Science, Vilnius
    University, Lithuania}
    \thanks{$^{\phantom{*}}$Email: stjukna@gmail.com,\ homepage: \href{http://www.thi.cs.uni-frankfurt.de/~jukna/}
    {http://www.thi.cs.uni-frankfurt.de/$\sim$jukna/}}
    \thanks{$^{\phantom{*}}$Submitted to: \emph{SIAM J. Discrete Math.}}

  \begin{abstract}
    The problem of constructing hazard-free Boolean circuits (those
    avoiding electronic glitches) dates back to the 1940s and is an
    important problem in circuit design and even in cybersecurity. We show that a DeMorgan circuit is hazard-free if and only if the
    circuit produces (purely syntactically) all prime implicants as
    well as all prime implicates of the Boolean function it computes.
    This extends to arbitrary DeMorgan circuits a classical result of
    Eichelberger [\emph{IBM J. Res. Develop.,} 9 (1965)] showing this
    property for special depth-two circuits. Via an amazingly simple
    proof, we also strengthen a recent result Ikenmeyer et
    al. [\emph{J. ACM}, 66:4 (2019)]: not only the complexities of
    hazard-free and monotone circuits for monotone Boolean functions
    do coincide, but every optimal hazard-free circuit for a monotone
    Boolean function must be monotone. Then we show that hazard-free
    circuit complexity of a very simple (non-monotone) Boolean
    function is super-polynomially larger than its unrestricted
    circuit complexity. This function accepts a Boolean $n\times n$
    matrix iff every row and every column has exactly one
    $1$-entry. Finally, we show that every Boolean function of $n$
    variables can be computed by a hazard-free circuit of
    size~$O(2^n/n)$.
  \end{abstract}

\maketitle

\keywords{\footnotesize {\bf Keywords:} Hazard-free circuits, monotone circuits, lower bounds}


\section{Introduction}

The problem of designing hazard-free circuits naturally occurs when
implementing circuits in hardware (\cite{Cadwell,huffman}, but is also
closely related to questions in logic
(\cite{kleene,koerner,malinowski}) and even in cybersecurity
(\cite{hu,tiwari}). The importance of hazard-free circuits is already
highlighted in the classical textbook~\cite{Cadwell}.

In this paper, under a \emph{circuit} we will understand a
\emph{DeMorgan circuit}, that is, a Boolean circuit with AND, OR and
NOT operations as gates, where negations are applied only to input
variables.
Inputs are constants $0$ and $1$, variables $x_1,\ldots,x_n$ and
their negations $\bar{x}_1,\ldots,\bar{x}_n$.  A~\emph{monotone}
circuit is a DeMorgan circuit without negated inputs. The \emph{size}
of a circuit is the number of non-input gates, and the \emph{depth} is the maximum, over all input-to-output paths, of the number of wires in these paths.

Roughly speaking, hazards are spurious pulses or electronic glitches
that may occur at the output of a circuit during an input transition,
stipulated by physical delays at wires or gates in a specific hardware
implementation of the circuit. Having designed a hazard-free circuit
for a given Boolean function, one is sure that no glitches will occur
in \emph{any} hardware implementation of this circuit, regardless of
the physical delays.

 A mathematical definition of hazards not relating on a vague notion of
``possible glitches'' is given in \cref{sec:hazards}. But to give an
illuminating example right now, let us consider an optimal circuit
$F=xz\lor y\bar{z}$ for the Boolean function which outputs $x$ if
$z=1$ and outputs $y$ if $z=0$ (this function is known as
\emph{multiplexer}). If due to different delays at wires or gates, the
output of the AND gate $xz$ of the circuit $F$ reaches the output OR
gate \emph{later} than that of the AND gate $y\bar{z}$, and if we
replace input $a=(1,1,0)$ by $b=(1,1,1)$, then the circuit will output
$0$ before it outputs the correct value~$1$: for a short moment, the
output OR gate will see the \emph{new} value $0$ of $y\bar{z}$ on the
input $b$, and the \emph{old} value $0$ of $xz$ on the input~$a$.
Thus, some hardware implementations of the circuit $F$ may have a
spurious $1$-$0$-$1$ glitch during the input transition $a\to b$.

That every Boolean function $f$ can be computed by a hazard-free
DeMorgan circuit was shown by Huffman~\cite{huffman} already in 1957:
the DNF whose terms are prime implicants of $f$ is hazard-free.
In 1965, Eichelberger~\cite[Theorem~2]{Eichelberger} extended
this result: if a DNF $D$ representing $f$ contains
no zero-terms, that is, terms with a variable together with its
negation, then $D$ is hazard-free if and only if $D$ contains all
prime implicants of~$f$ as terms. After these structural results,
research mainly concentrated on developing algorithms for detecting
hazards, and many important results were obtained~\cite{beredeson72,Eichelberger,mate,yoeli}, just to mention some of them.

But somewhat surprisingly, the following natural question remained
open: by how much must the size of a circuit be increased to ensure
hazard-freeness?  Only recently, an important progress towards this
question was made by Ikenmeyer et al.~\cite{Ikenmeyer}: the
hazard-free circuit complexity of \emph{monotone} Boolean functions
coincides with their monotone circuit complexity. Together with known
lower bounds on the later complexity, this shows that achieving
hazard-freeness \emph{can} require a super-polynomial blow up in circuit
size. Moreover, Ikenmeyer et al.~\cite{Ikenmeyer}, and Komarath and
Saurabh~\cite{komarath} have shown that detecting the presence of
hazards is computationally hard.

Our results are the following.

\begin{enumerate}

\item A  DeMorgan circuit computing a Boolean function~$f$ is hazard-free if and only if it produces (purely syntactically) all prime implicants and all prime implicates of~$f$ (\cref{cor:ext-Eichenberger} proved in \cref{sec:dual}). This removes the ``no zero-terms''  requirement in Eichelberger's theorem~\cite[Theorem~2]{Eichelberger}.

\item Every minimal hazard-free circuit computing a monotone Boolean function is monotone; this is a direct consequence of  \cref{thm1} proved in \cref{sec:monot}. It strengthens  the aforementioned result of Ikenmeyer et al.~\cite{Ikenmeyer}, and the proof is surprisingly simple.

\item Already very simple Boolean functions require a supper-polynomial blow up in the circuit size and depth to be computed by hazard-free circuits (\cref{thm3,thm3b} proved in \cref{sec:gaps}).

\item Every Boolean function of $n$ variables can be computed by a hazard-free circuit of size~$O(2^n/n)$ (\cref{sec:final}).
\end{enumerate}
The paper is organized as follows. In the next section, we recall standard concepts concerning Boolean functions and circuit, as well as of hazards. \Cref{sec:monot,sec:complexity,sec:gaps} concern the complexity of hazard-free circuits (results 2 and 3); the proofs here are very simple.
In \cref{sec:struct}, we give four necessary and sufficient conditions for a DeMorgan circuit to be hazard-free (implying result 1). Along the way, we show the duality between $0$-hazards and $1$-hazards (\cref{sec:dual}). In the last \cref{sec:final}, we show that every Boolean function of $n$ variables can be computed by a hazard-free circuit of size~$O(2^n/n)$ (result~4), and state open problems.

\section{Preliminaries}
\label{sec:hazards}
We will use standard concepts concerning Boolean functions and
circuits as, for example, in the standard references~\cite{hammer,wegener}, but let us recall them for completeness.

 A \emph{literal} is
either a variable $x_i=x_i^{1}$ or its negation $\bar{x}_i=x_i^{0}$.
A \emph{term} is an AND of literals, and a \emph{clause} is an OR of
literals. A term is a \emph{zero-term} if it contains a variable $x_i$
together with its negation $\bar{x}_i$. Terms $0$ and $1$ are constant terms (note that $0$ is not a zero-term, it is a constant term). By analogy with the usual
convention for products, we often omit the operator $\land$ and denote
conjunction by mere juxtaposition. For example, we will write
$x\bar{y}z$ instead of $x\land \bar{y}\land z$.  Similarly, a
\emph{clause} is an OR of literals. A clause containing a variable
together with its negation is a \emph{one-clause}.  A~\emph{non-zero}
term is a term which is not a zero-term, and a \emph{non-one} clause
is a clause which is not a one-clause.  A \emph{DNF} (disjunctive
normal form) is an OR of terms, and a \emph{CNF} (conjunctive normal
form) is an AND of clauses.

An \emph{implicant} of a Boolean function $f:\{0,1\}^n\to\{0,1\}$ is a
non-zero term $t$ such that $t\leq f$ holds, that is, for every $a\in
\{0,1\}^n$, $t(a)=1$ implies $f(a)=1$. In other words, a non-zero term
$t$ is an implicant of $f$ if every evaluation of the literals of $t$
to $1$ already forces the function $f$ to take value $1$, regardless
of the $0/1$ values given to the remaining variables.  An implicant
$t$ of $f$ is a \emph{prime implicant} of $f$ if no proper subterm $t'$ of $t$ has this property, that is, if $t\leq t'$ and $t'\leq f$, then $t'=t$.  For example, if $f(x,y,z)=xy\lor x\bar{y}z$,
then $xy$, $x\bar{y}z$ and $xz$ are implicants of $f$, but $x\bar{y}z$
is not a prime implicant.  Dually, an \emph{implicate} of a Boolean
function $f:\{0,1\}^n\to\{0,1\}$ is a non-one clause $\cl$ such that
$f\leq c$ holds, that is, for every $a\in \{0,1\}^n$, $\cl(a)=0$
implies $f(a)=0$. An implicate $\cl$ of $f$ is a \emph{prime
  implicate} of $f$ if no proper subclause of $\cl$ has this
property. For example, if $f(x,y,z)=xz\lor y\bar{z}$, then $\cl=x\lor
y$ is a prime implicate of~$f$. Implicants and implicates of $f$
describe subcubes of the binary cube $\{0,1\}^n$ on which $f$ is
constant.

\subsection{Formal DNFs and CNFs}
Every DeMorgan circuit $F$ not only computes a unique Boolean
function,  but also
\emph{produces} (purely syntactically) a unique set of terms
as well as a unique set of clauses in a natural way.

Namely, at an input gate holding a literal or a constant, the unique produced term
and clause is this literal or constant.
The set of terms produced at an OR gate is
the union of the sets of terms produced at two gates entering this OR
gate. The set of terms produced at an AND gate consists of all terms
of the form $t=t_1\land t_2$, where $t_1$ is a term produced at one
of the two gates entering this AND gate, and $t_2$ is a term produced at the other
entering gate. The OR  of all terms produced at the output gate of $F$ is the \emph{formal DNF} of $F$

The set of clauses produced at the gates of a given circuit is defined dually by
interchanging the roles of OR and AND gates. Namely, the set of
clauses produced at an AND gate is the union of the sets of clauses
produced at the two gates entering this AND gate. The set of clauses
produced at an OR gate consists of all clauses of the form
$\cl=\cl_1\lor \cl_2$, where $\cl_1$ is a clause produced at one of the two gates
entering this OR gate, and $\cl_2$ is a clause produced at the other
entering gate. The AND  of all clauses produced
at the output gate of $F$ is the \emph{formal CNF} of $F$.

Let us stress and important point: when forming the formal DNFs or formal CNFs, all Boolean laws, except two, can be used to simplify the resulting formulas. The two exceptions are the annihilation laws $x\land\bar{x}=0$ and $x\lor\bar{x}=1$: they are \emph{not} used!
Thus, some produced terms may be zero-terms, and some produced clauses
may be one-clauses, that is, may contain a variable together with its
negation. For example, the formal DNF $xy\lor xz\lor y\bar{z}\lor
z\bar{z}$ produced by the circuit $F=(x\lor\bar{z})(y\lor z)$ contains
a zero-term $z\bar{z}$. These ``redundant'' terms and clauses have no
influence on the Boolean function computed by the circuit, but are
decisive in the context of hazards.

\subsection{Ternary logic}

In this paper, we ignore the electronic aspect of hazards, and stick
on their idealized, mathematical model as, for example, in
\cite{beredeson72,Eichelberger,yoeli,Ikenmeyer,komarath}.
The classical Kleene's three-valued ``strong logic of
indeterminacy''~\cite{kleene} extends the Boolean operations AND, OR
and NOT from the Boolean domain $\{0,1\}$ to the ternary domain
$\{0,\unc,1\}$, where the bits $0$ and $1$ are interpreted as
\emph{stable}, and the bit $\unc$ as \emph{unstable}:

\begin{equation}\label{eq:truth-tables}
  \mbox{
    \begin{tabular}{c|ccc}
      and & 0 & $\unc$ & 1\\
      \midrule
      0 & 0 & 0 & 0\\
      $\unc$ & 0 & $\unc$ & $\unc$\\
      1 & 0 & $\unc$ & 1
    \end{tabular}
    \qquad\qquad
    \begin{tabular}{c|ccc}
      or & 0 & $\unc$ & 1\\
      \midrule
      0 & 0 & $\unc$ & 1\\
      $\unc$ & $\unc$ & $\unc$ & 1\\
      1 & 1 & 1 & 1
    \end{tabular}
    \qquad\qquad
    \begin{tabular}{c|ccc}
      not & 0 & $\unc$ & 1\\
      \midrule
      & 1 & $\unc$ & 0
    \end{tabular}
  }
\end{equation}
Note that, if we define the unstable bit as $\unc=\tfrac{1}{2}$, then
these ternary operations turn into tropical operations: $x\land y
=\min(x,y)$, $x\lor y=\max(x,y)$ and $\bar{x}=1-x$.  It is easy to
verify that the system $(\{0,\unc,1\},\lor,\land)$ forms a
distributive lattice with zero element $0$ and universal element $1$;
see, for example, Yoeli and Rinon \cite{yoeli}. In particular, the
operations $\lor$ and $\land$ are associative, commutative and each distributes over the other.
Moreover, the elements $0$ and $1$ satisfy for every
$x\in\{0,\unc,1\}$: $x\land 0=0$, $x\land 1 = x$, $x\lor 0=x$ and
$x\lor 1=1$.  The absorbtion laws $x\lor xy=x$ and $x(x\lor y)=x$ as
well as the rules of de Morgan $\overline{x\lor
  y}=\bar{x}\land\bar{y}$ and $\overline{x\land y}=\bar{x}\lor\bar{y}$
also hold. The only difference from the Boolean algebra is that, over
the ternary domain~$\{0,\unc,1\}$, the annihilation laws
$x\land\bar{x}=0$ and $x\lor\bar{x}=1$ \emph{do not} hold: $0\land
\unc=0$ but $\unc\land\bar{\unc}=\unc\neq 0$, and $1\lor \unc=1$ but
$\unc\lor\bar{\unc}=\unc\neq 1$.

\begin{rem}
Since the annihilation laws $x\land\bar{x}=0$ and $x\lor\bar{x}=1$ are
\emph{not} used when constructing the formal DNF $D$ or the formal CNF
$C$ produced by a given circuit $F$, the ternary functions computed by
$D$ and $C$ coincide with that computed by the circuit $F$, that is,
$D(\tern)=C(\tern)=F(\tern)$ holds for every ternary vector
$\tern\in\{0,\unc,1\}^n$. This is a simple but important observation
which allows us to analyze the properties of ternary functions
$F:\{0,\unc,1\}^n\to\{0,\unc,1\}$ defined by DeMorgan circuits $F$ by
analyzing the properties of the formal DNFs and formal CNFs of these circuits.
\end{rem}

\begin{rem}\label{rem:mukaidono}
  One can equip the set $\{0,\unc,1\}^n$ of ternary vectors with a
  partial order $\Leq$, where $\tern\Leq\ttern$ means that the vector
  $\ttern$ is obtained from $\tern$ by replacing some unstable bits
  $\unc$ by stable bits $0$ or $1$.  Since the extensions \cref
  {eq:truth-tables} of gates AND, OR and NOT to the ternary domain
  $\{0,\unc,1\}$ are monotone with respect to $\Leq$, the function
  $F:\{0,\unc,1\}^n\to\{0,\unc,1\}$ computed by a DeMorgan circuit $F$
  is monotone with respect to $\Leq$. In particular, if $\tern\Leq
  \ttern$ and $F(\ttern)=\unc$, then also $F(\tern)=\unc$, and if $F(\tern)=\bit$ for a stable bit $\bit\in\{0,1\}$, then also $F(\ttern)=\bit$.
\end{rem}

\subsection{Hazards}
After the functions AND, OR, NOT computed at individual gates are
extended from the binary domain $\{0,1\}$ to the ternary domain
$\{0,\unc,1\}$ using the truth-tables \cref{eq:truth-tables}, every
DeMorgan  circuit $F$ computing a given Boolean
function $f:\{0,1\}^n\to\{0,1\}$ turns into a circuit computing some
ternary function $F:\{0,\unc,1\}^n\to\{0,\unc,1\}$, a ``ternary
extension'' of $f$, which coincides with $f$ on $\{0,1\}^n$. Even if
two circuits compute the same Boolean function, their ternary
extensions may be different. Whether a circuit $F$ is hazard-free or
not depends entirely on the properties of its ternary extension which,
in turn, depends on the specific form of the circuit~$F$. Namely, the
circuit $F$ has a hazard at some vector $\tern\in \{0,\unc,1\}^n$ if
the Boolean function $f$ computed by $F$ does not depend on the
unstable bits of~$\tern$, but still $F$ outputs the unstable bit
$\unc$ on the input~$\tern$. To be more specific, let us fix our
notation.

A \emph{resolution} of a ternary vector $\tern\in \{0,\unc,1\}^n$ is a
vector in $\{0,1\}^n$ obtained from~$\tern$ by replacing every
occurrence of the unstable bit $\unc$ by a stable bit $0$ or~$1$. The
\emph{subcube} defined by $\tern$ is the set
\[
\subcube=\big\{a\in\{0,1\}^n\colon \mbox{$a$ is a resolution of
  $\tern$}\big\}
\]
of all resolutions of~$\tern$; hence, $|\subcube|=2^m$, where $m$ is
the number of unstable bits $\unc$ in~$\tern$. If $\tern\in\{0,1\}^n$
(there are no unstable bits at all), then we let $\subcube=\{\tern\}$.
For a Boolean function $f$ and a set $A\subseteq \{0,1\}^n$, let $
f(A)=\{f(a)\colon a\in A\}\subseteq \{0,1\} $ denote the set of values
taken by $f$ on $A$. Hence, $f(A)=\{0\}$ iff $f(a)=0$ for all $a\in
A$, and $f(A)=\{1\}$ iff $f(a)=1$ for all $a\in A$. In particular,
$|f(A)|=1$ means that the function $f$ is constant on~$A$.

\begin{rem}\label{rem:no-switch}
If $F$ is a DeMorgan circuit, and $\tern\in\{0,\unc,1\}^n$, then
$F(\subcube)=\{\bit\}$ for $\bit\in\{0,1\}$ implies $F(\tern)\neq \bar{\bit}$, that is,  $F(\tern)\in\{\bit,\unc\}$.
This follows from \cref{rem:mukaidono}, but can be also shown directly. Suppose that $F(\subcube)=\{0\}$ but $F(\tern)=1$. Then $t(\tern)=1$ must hold for at least one term $t$ produced by $F$. Since $1\land\unc=\unc\neq 1$, this means that the vector $\tern$ evaluates to $1$ all literals of $t$.  But then
also every resolution of $a\in\subcube$ of $\tern$ evaluates these
literals to $1$, and we obtain $F(a)=t(a)=1$, a contradiction with
$F(a)=0$. If $F(\subcube)=\{1\}$, then $F(\tern)\in\{1,\unc\}$ follows by considering the clauses produced by~$F$.
\end{rem}

There are
several types of hazards---we will only consider the so-called
\emph{static logical hazards} \cite{beredeson72,Eichelberger,yoeli}.

\begin{dfn}[Hazards]\rm
  A circuit $F$ of $n$ variables has a \emph{hazard} at
  $\tern\in\{0,\unc,1\}^n$ if $F$ is constant on the subcube $\subcube$ but $F(\tern)=\unc$ holds. This is a $0$-\emph{hazard} if  $F(\subcube)=\{0\}$, and is a $1$-\emph{hazard} if $F(\subcube)=\{1\}$.
   A circuit is
  \emph{hazard-free} if it has a hazard at none of the inputs
  $\tern\in\{0,\unc,1\}^n$.
\end{dfn}

\begin{ex}\label{ex:no-haz}
Trivial examples of circuits with hazards are the two circuits $x\land\bar{x}$ and $x\lor\bar{x}$ computing the two constant functions $0$ and $1$. A less trivial example of such a circuit is  $F=xz\lor y\bar{z}$ computing the multiplexer function
$f(x,y,z)=xz\lor y\bar{z}$; we already considered this circuit in the introduction to give an intuition behind the hazards. This circuit has  a $1$-hazard at $\tern=(1,1,\unc)$:
  $F(\tern)=\unc\lor\bar{\unc}=\unc$ even though
  $f(1,1,0)=f(1,1,1)=1$. The circuit
  $H=(x\lor \bar{z})(y\lor z)$ for the same function has a $0$-hazard at $\tern=(0,0,\unc)$:
  $H(\tern)=\bar{\unc}\land\unc=\unc$ even though
  $f(0,0,0)=f(0,0,1)=0$.
\end{ex}

\begin{rem}
Let us stress that the mere fact that a circuit $F$ outputs $\unc$ on some $\tern\in\{0,\unc,1\}^n$ does not mean that $F$ has a hazard at $\tern$: for the latter to happen, the circuit must take the \emph{same} value on the entire Boolean subcube $\subcube$, that is, $|F(\subcube)|=1$ must also hold. In particular, if every two vectors in $f^{-1}(\bit)$ differ in at least two positions, then every DeMorgan circuit computing $f$ is free from $\bit$-hazards per se: if $\tern$ contains at least one $\unc$, then $F(\subcube)\neq\{\bit\}$.
So, for example, every circuit $F$ computing the parity function  $x_1\oplus\cdots\oplus x_n$ is hazard-free.
\end{rem}

A ``folklore'' observation is that $0$-hazards can only be introduced by zero-terms, and $1$-hazards can only be introduced by one-clauses. Recall that a circuit is \emph{monotone} if it has no negated variables as inputs; note that such circuits cannot produce any zero-terms or one-clauses.

\begin{prop}\label{prop:zero-terms}
  If a circuit $F$ produces no zero-terms, then $F$ has no
  $0$-hazards, and  if  $F$ produces no one-clauses, then $F$ has no
  $1$-hazards. In particular, monotone circuits are hazard-free.
\end{prop}

\begin{proof}
  Let $D$ be the formal DNF of the circuit $F$, and suppose that all
  terms of~$D$ are non-zero terms. Assume to the contrary that the circuit
  $F$ has a $0$-hazard at some vector $\tern\in\{0,\unc,1\}^n$; hence,
  $F(\subcube)=\{0\}$ but $F(\tern)=\unc$. Since
  $D(\tern)=F(\tern)=\unc$ and $1\lor\unc=1\neq\unc$, there must be a term $t$ in $D$ with
  $t(\tern)=\unc$. Since $0\land \unc=0\neq\unc$, the vector $\tern$
  evaluates every literal of $t$ either to $1$ or to $\unc$.  Since
  $t$ has no variable together with its negation, we can
  evaluate every literal of $t$ to $1$. On every such resolution
  $a\in\subcube$ of $\tern$, we have $t(a)=1$ and, hence, also
  $F(a)=D(a)=1$, a contradiction with $F(\subcube)=\{0\}$.
  The proof of the second claim (for $1$-hazards) is dual by considering the formal CNF of~$F$.
\end{proof}

\begin{ex}\rm\label{ex:basic}
  Consider the circuit $F=x(y\lor z)\lor y\bar{z}$ computing the multiplexer
  function $f(x,y,z)=xz\lor y\bar{z}$ from \cref{ex:no-haz}. Since the circuit $F$ produces no zero-terms, it has no $0$-hazards, by \cref{prop:zero-terms}. On the other hand,
   the function $f$ can take value $1$ only if $x=z=1$ or $y=\bar{z}=1$ or $x=y=1$. But $F(1,\unc,1)=1\lor \unc=1\neq\unc$, $F(\unc,1,0)=\unc\lor 1=1\neq\unc$ and $F(1,1,\unc)=1\lor \unc=1\neq\unc$. So, $F$ has no $1$-hazards as well.
\end{ex}

\section{Hazard-free and monotone circuits}
\label{sec:monot}
Recall that a Boolean function $f:\{0,1\}^n\to\{0,1\}$ is
\emph{monotone} if $f(x)=1$ and $x\leq y$ imply $f(y)=1$, where $x\leq
y$ means that $x_i\leq y_i$ holds for all positions~$i$.  The
\emph{upwards closure} of a not necessarily monotone Boolean function
$f(x)$ is the monotone Boolean function
\[
\up{f}(x):=\bigvee_{z\leq x}f(z)\,.
\]
That is, $\up{f}(x)=1$ iff $f(z)=1$ holds for some vector $z\leq x$;
we use the term ``upwards'' because if $f$ accepts a vector $z$, then
$\up{f}$ accepts all ``larger'' vectors $x\geq z$. For example, if
$f(x)=x_1\oplus\cdots\oplus x_n$ is the parity function, then
$\up{f}(x)=x_1\lor\cdots\lor x_n$: if $x_i=1$ holds for at least one
$i$, then $f(z)=1$ for the vector $z\leq x$ with exactly one $1$ in
the $i$th position. Note that
$\up{f}=f$ holds for all monotone Boolean functions~$f$.
Let us also note that monotone Boolean functions corresponding to many NP-hard problems are upwards closures of relatively simple non-monotone Boolean functions. Consider, for example, a Boolean function $f(x)$ of
 $n=\binom{m}{2}$ variables such that
 $f(x)=1$ iff the $m$-vertex graph $G_x$ encoded by a $0$-$1$ vector $x$  consists of a complete graph on some $m/2$ vertices and $m/2$ isolated vertices. Then $\up{f}(x)=\ncliq(x)$  is the well known NP-complete clique function: $\ncliq(x)=1$ iff $G_x$ contains a complete subgraph on $m/2$ vertices.

We can view every DeMorgan circuit $F(x)$ computing a Boolean function
$f(x)$ of $n$ variables as a monotone circuit $H(x,y)$ on $2n$
variables with the property that $F(x)=H(x,\bar{x})$ holds for all
$x\in\{0,1\}^n$, where $\bar{x}=(\bar{x}_1,\ldots,\bar{x}_n)$ is the
complement of $x=(x_1,\ldots,x_n)$. The
\emph{monotone version} of the circuit $F(x)$ is the monotone circuit
$\pos{F}(x)=H(x,\vec{1})$ obtained by replacing every negated input
literal $\bar{x}_i$ with constant~$1$. For example, the
monotone version of the circuit $F=y\bar{z}\lor x(\bar{y}\lor \bar{x}y)$
is $\pos{F}= y\cdot 1\lor x(1\lor 1\cdot y)=x\lor y$.

\begin{rem}\label{rem:closure}
  Since the circuit $H(x,y)$ is monotone, and since the circuit
  $\pos{F}(x)=H(x,\vec{1})$ is obtained from $H$ by replacing with
  constant $1$ some of its inputs (namely, all negated input literals), we have $\pos{F}(x)\geq f(x)$.
  Since the circuit $\pos{F}(x)$ is monotone, we also have
  $\pos{F}(x)\geq \pos{F}(z)$ for every $z\leq x$. Thus,  $\pos{F}(x)\geq \up{f}(x)$ holds
  for all $x\in\{0,1\}^n$.
\end{rem}

\begin{thm}\label{thm1}
  Let $F$ be a DeMorgan circuit computing a Boolean function $f$.  If
  the circuit~$F$ has no $0$-hazards, then $\pos{F}$
  computes~$\up{f}$.
\end{thm}

\begin{proof}
  Assume that the circuit $\pos{F}$ does not compute~$\up{f}$. Then,
  by \cref{rem:closure}, there must be a vector $a\in\{0,1\}^n$ such
  that $\pos{F}(a)=1$ but $\up{f}(a)=0$.  The formal DNF $\pos{D}$ of the circuit $\pos{F}$ is obtained by replacing
  with constant $1$ every negated literal in the formal DNF $D$ of the circuit~$F$. Since $\pos{D}(a)=\pos{F}(a)=1$,
  there must be a term $t=\bigwedge_{i\in A}x_i\land \bigwedge_{i\in
    B}\bar{x}_i$ in the DNF~$D$ such that $\pos{t}(a)=1$ holds for its
  subterm $\pos{t}=\bigwedge_{i\in A}x_i$. From $D(a)=f(a)=0$, we have $t(a)=0$; hence, the set $B'=\{i\in B\colon a_i=1\}$ is nonempty.
 Take the ternary vector
  $\tern\in\{0,\unc,1\}^n$ with $\tern_i=\unc$ for all $i\in B'$, and
  $\tern_i=a_i$ otherwise.

  On this vector, we have $t(\tern)=1\land \unc=\unc$. Since the Boolean vector $a$ evaluates every term of $D$ to $0$, the ternary
   vector $\tern$ evaluates  every other term of $D$ to either $0$ or $\unc$. Hence, $F(\tern)=D(\tern)=\unc$.
   On the other
  hand, since the ternary vector $\tern$ has unstable bits $\unc$ only
  in positions where the binary vector $a$ has $1$s, every resolution
  $b\in\subcube$ of $\tern$ satisfies $b\leq a$. Since $\bigvee_{b\leq
    a}f(b)=\up{f}(a)=0$, we have  $f(b)=0$ for every
  resolution $b\in\subcube$ of $\tern$, meaning that
  $f(\subcube)=\{0\}$. Thus, the circuit $F$ has a $0$-hazard at
  $\alpha$.
\end{proof}

The following example shows that the converse of \cref{thm1} does not hold: the monotone version
  $\pos{F}$ of a circuit $F$ may compute $\up{f}$ even though the
  circuit $F$ has $0$-hazards as well as $1$-hazards.

\begin{ex}\label{ex:thm1}
  Consider the
  circuit $F=y\bar{z}\lor x(\bar{y}\lor \bar{x}y)$ computing the
  Boolean function $f(x,y,z)= x\bar{y}\lor y\bar{z}$. The upwards
  closure of $f$ is $\up{f}=x\lor y$. The monotone version $\pos{F}=
  y\cdot 1\lor x(1\lor 1\cdot y)=x\lor y$ of the circuit $F$ computes
  $\up{f}$. But on the input vector $\tern=(\unc,1,1)$, we have
  $F(\tern)= 0\lor \unc(0\lor \bar{\unc})=\unc\,\bar{\unc}=\unc$ even
  though $f(0,1,1)=f(1,1,1)=0$, while on the vector
  $\ttern=(1,\unc,0)$, we have $F(\ttern)=\unc \lor (\bar{\unc}\lor
  0)=\unc\lor\bar{\unc}=\unc$ even though $f(1,0,0)=f(1,1,0)=1$.
\end{ex}

By \cref{thm1}, hazard-freeness of a circuit $F$ computing a Boolean
function $f$ is a sufficient condition for $\pos{F}=\up{f}$ to
hold, but \cref{ex:thm1} shows that this condition is not
necessary. The following theorem shows what is actually necessary.
The \emph{positive factor} $\pos{t}$ of a term
$t$ is obtained by replacing every its negated literal with
constant~$1$.

\begin{thm}\label{thm2}
  Let $F$ be a DeMorgan circuit computing a Boolean function $f$. Then
  the circuit $\pos{F}$ computes $\up{f}$ if and only if the positive
  factor of every zero-term produced by $F$ is an implicant
  of~$\up{f}$.
\end{thm}

In particular, $\pos{F}=\up{f}$ if $F$ produces no zero-terms.

\begin{proof}
  Let $D=\bigvee_{t\in T}t$ be the formal DNF of $F$. Then the formal
  DNF of the monotone version $\pos{F}$ of $F$ is the OR
  $\pos{D}=\bigvee_{t\in T}\pos{t}$ of positive factors of terms
  of~$D$.  Now take an arbitrary term $t=\bigwedge_{i\in
    A}x_i\land\bigwedge_{j\in B}\bar{x}_j$ of $D$. If $t$ is a
  zero-term ($A\cap B\neq\emptyset$), then it contains a subterm
  $x_i\bar{x}_i$ for some $i$. Since $x_i\bar{x}_i(b)=0$ holds for all
  vectors $b\in\{0,1\}^n$, $\up{t}(a)=0$ holds for all $a$ as well.
  Now suppose that $t$ is a non-zero term ($A\cap B=\emptyset$), and
  take an arbitrary vector $a\in\{0,1\}^n$. If $\up{t}(a)=1$, then
  $t(b)=1$ for some $b\leq a$. From $t(b)=1$ we have $b_i=1$ for all
  $i\in A$, and from $b\leq a$, we also have $a_i=1$ for all $i\in A$,
  that is, $\pos{t}(a)=1$. If $\pos{t}(a)=1$, then $t(b)=1$ holds for
  the vector $b\leq a$ with $b_i=a_i$ for all $i\in A$ and $b_i=0$ for
  all $i\not\in A$, and $\up{t}(a)=t(b)=1$ follows.

  Thus, for every term $t$ we have $\up{t}=0$ if $t$ is a zero-term,
  and $\up{t}=\pos{t}$ if $t$ is a non-zero term. Since $\up{(g\lor
    h)}=\up{g}\lor\up{h}$ holds for any Boolean functions
  $g,h:\{0,1\}^n\to\{0,1\}$, we obtain
  \begin{equation}\label{eq:posDNF}
    \up{f}=\bigvee_{t\in T}\up{t}=\bigvee_{t\in T'}\pos{t}\leq \bigvee_{t\in T}\pos{t}=\pos{F}\,,
  \end{equation}
  where $T'\subseteq T$ is the set of all non-zero terms of~$D$. By
  \cref{eq:posDNF}, the equality $\up{f}=\pos{F}$ holds if and only if
  $\pos{t}\leq \up{f}$ holds for every term $t\in T\setminus T'$, that
  is, if and only if the positive factor $\pos{t}$ of every zero-term
  $t$ of $D$ is an implicant of~$\up{f}$.
\end{proof}

 By \cref{prop:zero-terms}, if a circuit $F$ has a $0$-hazard, then $F$ must produce at least one zero-term.  \Cref{thm1,thm2} yield a partial converse: if $F$ produces a zero-term whose positive factor is not an implicant of the upwards closure of the Boolean function computed by $F$, then $F$ has a $0$-hazard.

\begin{cor}\label{cor:sufficient}
Let $F$ be a DeMorgan circuit computing a Boolean function~$f$.
If~$F$ produces a zero-term $t$ such that $\pos{t}\not\leq \up{f}$, then $F$ has a $0$-hazard.
\end{cor}

\Cref{ex:thm1} shows that the converse of \cref{cor:sufficient} does not hold.

\section{Complexity bounds}
\label{sec:complexity}

For a Boolean function $f$, let $\uL{f}$ denote the minimum number of
gates in a DeMorgan circuit computing~$f$. By $\hL{f}$ and $\mL{f}$ we
denote the versions of this measure when restricted, respectively, to
hazard-free circuits and to monotone circuits. Finally, let $\zL{f}$
denote the minimum number of gates in a DeMorgan circuit that computes $f$ and produces no zero-terms. In the measure $\zpL{f}$, we
additionally require that the circuit must produce all prime implicants
of~$f$; hence, $\zL{f}\leq \zpL{f}$. We only state the following corollary for the circuit \emph{size} measures, but it also holds for the circuit \emph{depth} measures.

\begin{cor}\label{cor:lowerA}
  For every Boolean function $f$, we have
  \[
  \mL{\up{f}}\leq \hL{f}\leq \zpL{f}\qquad \mbox{and}\qquad \mL{\up{f}}\leq \zL{f}\,.
   \]
  If $f$ is monotone, then
  \[
  \mL{f}=\hL{f}=\zL{f}=\zpL{f}\,.
  \]
\end{cor}

\begin{proof}
  The inequality $\mL{\up{f}}\leq \hL{f}$  follows from \cref{thm1}, and the inequality $\mL{\up{f}}\leq \zL{f}$ follows from \cref{thm2}.
  The inequality $\hL{f}\leq \zpL{f}$ follows from a classical result of Eichelberger~\cite[Theorem~2]{Eichelberger}: if a DeMorgan circuit
  $F$ computing a Boolean function $f$ produces no zero-terms, then $F$ is hazard-free if and only if the circuit~$F$ produces all prime implicants of~$f$.  Now let $f$ be a monotone Boolean function; hence, $\up{f}=f$.
  To show that then $\mL{f}=\hL{f}=\zL{f}=\zpL{f}$ holds, it is enough to  show that  $\zpL{f}\leq \mL{f}$
  holds. So, let $F$ be a monotone circuit of size $\mL{f}$
  computing~$f$. Since the circuit $F$ has no negated inputs,
  it cannot produce any zero-terms, and it is enough to show that
  every prime implicant of $f$ must be produced by~$F$. This a well
  known and easy to verify fact.

  Assume for a contradiction that some prime implicant
  $p=\bigwedge_{i\in S}x_i$ of~$f$ is not produced by $F$, and
  consider the vector $a\in\{0,1\}^n$ with $a_i=1$ for all $i\in S$
  and $a_i=0$ for all $i\in S$. On this vector, we have $f(a)=p(a)=1$.
  But since every term $t\neq p$ produced by $F$ must be an implicant
  of $f$, and since $p$ is a \emph{prime} implicant, $t$ must have a
  variable $x_i$ with $i\not\in S$. Thus, $t(a)=0$ holds for all terms
  $t$ produced by $F$. But then $F(a)=0\neq f(a)$, a contradiction.
\end{proof}

The lower bound $\hL{f}\geq \mL{\up{f}}$
was already shown by Ikenmeyer et al.~\cite{Ikenmeyer} as a special case of a more general result proved using different  arguments.
 Associate with every Boolean vector $x\in\{0,1\}^n$ the ternary vector $\xx$ whose $i$th position is $a_i$ if $x_i=0$, and is $\unc$ if $x_i=1$. That is, vector $x$ tells us which bits of $a$ are changed from stable to unstable. The \emph{hazard derivative}
$\der{F}{a}{x}$ of a ternary function $F:\{0,\unc,1\}^n\to\{0,\unc,1\}$
computable by a DeMorgan circuit at a  point $a\in\{0,1\}^n$ is defined by letting $\der{F}{a}{x}=0$ if $F(\xx)=F(a)$, and $\der{F}{a}{x}=1$ if $F(\xx)=\unc$; by \cref{rem:mukaidono}, there are only these two possibilities.
 This concept is extended to Boolean functions $f$ by letting $\der{f}{a}{x}=0$ iff $f(a\oplus z)=f(a)$ for all
$z\leq x$; here, $x\oplus z$ is the componentwise {\sf xor} of Boolean vectors $x$ and $z$. The function $d(x)=\der{f}{a}{x}$ is clearly monotone: if $x\leq y$, and if something holds for all $z\leq y$, then this also holds for all $z\leq x$.

The core of the entire argument in ~\cite{Ikenmeyer} is a chain rule for the hazard derivatives $\der{F}{a}{x}$. The authors then use this rule to
transform a given hazard-free circuit $F$ for~$f(x)$ into a monotone
circuit computing the hazard derivatives $\der{f}{a}{x}$ of~$f$ at all
points~$a$. The argument is reminiscent of that used by Baur and
Strassen~\cite{strassen} to compute all partial derivatives of a
multivariate polynomial by an arithmetic circuit. This leads to
the lower bound $\uL{f} \geq \mL{\der{f}{a}{x}}$ for every $a\in\{0,1\}^n$. If $f(\vec{0})=0$, then taking
$a=\vec{0}$ we obtain $\der{f}{\vec{0}}{x}=0$ iff $f(z)=0$ for all
$z\leq x$, which happens precisely when $\up{f}(x)=0$. Thus
$\der{f}{\vec{0}}{x}=\bigvee_{z\leq x} f(x)=\up{f}(x)$. Note that if $f(\vec{0})=1$, then $\up{f}=1$, and the lower bound $\hL{f}\geq \mL{1}$ trivially holds.

\Cref{thm1} gives an alternative, short and direct proof of the lower
bound $\hL{f}\geq \mL{\up{f}}$ using the mere definition of hazards:
if $\pos{F}(a)\neq \up{f}(a)$, then the circuit $F$ must produce a
term $t$ such that $\pos{t}(a)=1$ but $t(a)=0$, and the circuit $F$ has a $0$-hazard at the ternary vector $\tern\in\{0,\unc,1\}^n$ with $\tern_i=\unc$ if $a_i=1$ and $\bar{x}_i\in t$, and $\tern_i=a_i$ otherwise.
Moreover, in this theorem, the desired monotone circuit $\pos{F}$ computing $\up{f}$ is obtained from a hazard-free circuit $F$ computing $f$ by just replacing all negated inputs of $F$ with constant~$1$: no further transformations of the circuit itself are necessary.
Thus, for monotone Boolean functions
$f$, \cref{thm1} tells us a bit more than the mere equality
$\hL{f}=\mL{f}$: it shows that not only hazard-free and monotone
circuit \emph{complexities} for monotone Boolean functions $f$ do
coincide but, in fact, that \emph{every} minimal hazard-free circuit for $f$ is a monotone circuit \emph{itself}, that is, does not use negated
input variables to compute its values.

\section{Complexity gaps}
\label{sec:gaps}
Together with already known lover bounds on the monotone circuit
complexity, the lower bound $\hL{f}\geq \mL{\up{f}}$ implies that the
gap $\hL{f}/\uL{f}$ can be super-polynomial and even exponential.
Such gaps were shown in~\cite{Ikenmeyer}, when $f$ is either the
logical permanent~\cite{razb-perm} or the logical determinant, or the
Tardos function~\cite{tardos}. However, the known circuits for these
functions demonstrating that $\uL{f}$ is polynomial are far from
being trivial. Actually, except for determinant~\cite{berkowitz}, we even do not have
\emph{explicit} constructions of these circuits---we only have general
algorithms: \cite{Kuhn55,hopcroft73} for logical permanent and
\cite{GLS81} for the Tardos function.

We now show that a super-polynomial gap $\hL{f}/\uL{f}$ is actually
achieved on a very simple \emph{exact perfect matching}
function~$f_n$ of $n=m^2$ variables. Inputs are Boolean  $m\times m$ matrices
$x=(x_{i,j})$, and $f_n(x)=1$ if and only if $x$ is permutation
matrix, that is, if every row and every column of $x$ has exactly
one~$1$. By viewing $x$ as the adjacency matrix of a bipartite $n\times n$ graph $G_x$, we have  $f_n(x)=1$ if and only if $G_x$ is a perfect matching.

\begin{thm}\label{thm3}
  The exact perfect matching function $f_n$ can be computed by circuit of size $O(n)$ and
  depth $O(\log n)$, but any hazard-free circuit computing $f_n$ must
  have size $n^{\Omega(\log n)}$ and depth~$\Omega(n)$.
\end{thm}

\begin{proof}
  The \emph{logical permanent} function $\uper$ accepts a Boolean
  $m\times m$ matrix $x$ iff $f_n(z)=1$ holds for at least one matrix $z\leq
  x$. Hence, $\uper=\up{f}_n$ is the upwards closure of $f_n$.
  Razborov~\cite{razb-perm} has shown that any monotone circuit
  computing $\uper$ must have size $n^{\Omega(\log n)}$, and Raz and
  Wigderson~\cite[Theorem~4.2]{RW92} have shown that any monotone
  circuit computing $\uper$ has depth $\Omega(n)$. Together with
  \cref{cor:lowerA}, this implies that any hazard-free circuit computing
  $f_n$ must have size $n^{\Omega(\log n)}$ and depth~$\Omega(n)$.
On the other hand, every \emph{exact-$k$} function
\[
\mbox{$\exk{k}(x_1,\ldots,x_m)=1$ if and only if $x_1+\cdots+x_m=k$,}
\]
which accepts an input vector iff it has exactly a given number $k$ of $1$s, is symmetric, and it is known (see, e.g. \cite[Chapter~3.4]{wegener}) that every symmetric Boolean function of $m$ variables can be computed by a DeMorgan circuit of size $O(m)$ and depth $O(\log m)$. So, the exact perfect matching
function~$f_n$  can be computed by a circuit of size $O(m^2)=O(n)$ and depth $O(\log n)$.
\end{proof}

\begin{rem}
By allowing slightly larger than linear number of gates,
the exact per\-mu\-ta\-tion function~$f_n$ can be directly computed by a trivial circuit without using circuits for symmetric functions. 
\ignore{Consider the circuits:
\begin{enumerate}
\item[] $R_{i,j}$ = AND of $x_{i,j}$ and all $\bar{x}_{i,l}$ for $l\neq j$;
\item[] $C_{i,j}$ = AND of $x_{i,j}$ and all $\bar{x}_{k,j}$ for $k\neq i$;
\item[] $R_i=R_{i,1}\lor R_{i,2}\lor\cdots\lor R_{i,n}$;
\item[] $C_j=C_{1,j}\lor C_{2,j}\lor\cdots\lor C_{n,j}$;
\item[] $F=R_1\land\cdots\land R_n\land C_1\land\cdots\land C_n$.
\end{enumerate}
Note that $R_{i,j}(x)=1$ iff the $i$th row of $x$ has exactly one $1$ in the $j$th column, and  $C_{i,j}(x)=1$ iff the $j$th column of $x$ has exactly one $1$ in the $i$th row. So, $R_i(x)=1$ iff the $i$th row of $x$ has exactly one $1$, and $C_j(x)=1$ iff the $j$th column of $x$ has exactly one $1$. Thus, the circuit $F$ (which, actually, is a formula) computes~$f_n$. The
  size of this circuit is $O(m^3)=O(n^{3/2})$ and the depth is $O(\log n)$. If unbounded fanin gates are allowed, when the size of $F$ is $O(m^2)=O(n)$ and the depth is~$3$.

.....................
}
Consider the circuit $F=F_1\land F_2$, where
  \[
  F_1=\bigwedge_{i=1}^m\bigvee_{j=1}^m x_{i,j}\land \bigwedge_{k\neq
    j}\bar{x}_{i,k}\ \mbox{ and }\
  F_2=\bigwedge_{j=1}^m\bigvee_{i=1}^m x_{i,j}\land \bigwedge_{l\neq
    i}\bar{x}_{l,j}\,.
  \]
  Note that $F_1(x)=1$ iff every row of $x$ has exactly one $1$, and
  $F_2(x)=1$ iff every column of $x$ has exactly one $1$, meaning that
  the circuit $F$ (which, actually, is a formula) computes~$f_n$. The
  size of this circuit is $O(m^3)=O(n^{3/2})$ and the depth is $O(\log n)$. If unbounded fanin gates are allowed, when the size of $F$ is $O(m^2)=O(n)$ and the depth is~$3$.
\end{rem}

By using less trivial circuits, one can increase the gap from
  super-polynomial to exponential.
The \emph{exact $k$-clique} function $\ecliq$ has $n=\binom{m}{2}$ variables corresponding to the edges of the complete graph $K_m$ on $m$ vertices. Every assignment $x$ of $0/1$ values to these variables specifies a subgraph $G_x$ of $K_m$, and $\ecliq(x)=1$ iff $G_x$ is an exact $k$-clique, that is, consists of a complete graph on some $k$ vertices and $m-k$ isolated vertices.

\begin{thm}\label{thm3b}
For every $1\leq k\leq m$, the exact $k$-clique function $\ecliq$ can be computed by a DeMorgan circuit of size $O(n)$ and depth $O(\log n)$, but for $k=\lfloor \tfrac{1}{4}(m/\log m)^{1/3}\rfloor$, any hazard-free circuit computing $\ecliq$ must
  have size exponential in $\Omega\left((n/\log n)^{1/6}\right)$ and depth $\Omega\left((n/\log n)^{1/6}\right)$.
\end{thm}

\begin{proof}
The upwards closure $\up{\ecliq}$ of $\ecliq$ is the well-known \emph{$k$-clique function}: $\cliq(x)=1$ iff $G_x$ contains a $k$-clique, that is, a complete subgraph on $k$ vertices. Alon and Boppana~\cite[Theorem~3.9]{AB87} have show that, for $k=\lfloor \tfrac{1}{4}(m/\log m)^{1/3}\rfloor$,  every monotone circuit computing $\cliq$ must have at least $\exp\left((n/\log n)^{1/6}\right)$ gates. For this choice of $k$, the lower bound proved by Goldmann and
H\aa stad~\cite[Theorem~3]{goldmann} implies that every monotone circuit for $\cliq$ must have depth at least a constant times $(n/\log n)^{1/6}$.
By \cref{cor:lowerA},
every hazard-free circuit computing $\ecliq$ must have at least so large size and depth.

To show the upper bounds, observe that a subgraph of $K_m$ on $\{1,\ldots,m\}$ is an exact $k$-clique iff it has exactly $k$ vertices of degree $k-1$ and $k(k-1)/2$ edges in total. We first compute the values  $y_i=\exk{k-1}(x_{i,1},\ldots,x_{i,i-1},x_{i,i+1},\ldots,x_{i,m})$ for all vertices $i$ of $K_m$. Since $y_i=1$ iff the vertex $i$ has degree $k-1$, the circuit $F(x)=\exk{k(k-1)/2}(x)\land \exk{k}(y_1,\ldots,y_m)$ computes $\ecliq$.
The size of the circuit $F$ is $O(m^2)=O(n)$ and the depth is $O(\log n)$.
\end{proof}

\begin{rem}\label{rem:negations}
An intuitive explanation for large gaps between the sizes of hazard-free and unrestricted circuits is
given by \cref{cor:sufficient}. In an unrestricted circuit $F$ computing the exact $k$-clique function $\ecliq$, we have no restrictions on the form of  produced zero-terms. However, if $F$ is hazard-free, then \cref{cor:sufficient} implies that every produced zero-term $t$ must have the following property: the graph encoded by the unnegated variables of $t$ must contain a $k$-clique. This is a severe restriction on the usage of negations which makes hazard-free circuits almost as weak as monotone circuits.
\end{rem}

\section{Structure of hazards}
\label{sec:struct}
A classical result of Eichelberger~\cite[Theorem~2]{Eichelberger} states that if a DNF representing a Boolean function $f$ has no zero-terms, then it is hazard-free if and only if it contains all prime implicants of $f$ as terms.   The goal of this section is to remove the ``no zero-terms produced'' restriction from Eichelberger's theorem, and to establish further structural properties of hazards.

\subsection{Ternary vectors as terms and clauses}
It will be convenient to identify implicants and implicates of boolean functions with ternary vectors.  Namely, associate with every ternary vector $\tern\in\{0,\unc,1\}^n$ the following term and clause:
\[
\terntr:= \bigwedge_{i: \tern_i\neq\unc} x_i^{\tern_i}\ \ \mbox{ and }\
\ \terncl:=\bigvee_{i: \tern_i\neq\unc}x_i^{1-\tern_i}\,.
\]
 For example, if
$\tern=(1,\unc,0,\unc)$, then $\terntr=x_1\bar{x}_3$ and
$\terncl=\bar{x}_1\lor x_3$.  Note that $\terntr(\tern)=1$ and $\terncl(\tern)=0$. So, the terms and clauses associated with vectors $\tern\in\{0,\unc,1\}^n$ define the subcube~$\subcube$:
\[
\terntr^{-1}(1)=\subcube=\terncl^{-1}(0)\,,
\]
where, as customary, $f^{-1}(\bit)=\left\{a\in\{0,1\}^n\colon f(a)=\bit \right\}$.
We say that a ternary vector $\tern\in\{0,\unc,1\}^n$ is a
$1$-\emph{witness} of a Boolean function $f$ if the term $\terntr$ is
an implicant of $f$, that is, if $\terntr\leq f$ holds, and $\tern$ is
a \emph{prime $1$-witness} of $f$ if the term $\terntr$ is such.
Dually, $\tern\in\{0,\unc,1\}^n$ is a $0$-\emph{witness} of $f$ if the
clause $\terncl$ is an implicate of $f$, that is, if $f\leq \terncl$
holds, and $\tern$ is a \emph{prime $0$-witness} of $f$ if the clause
$\terncl$ is such. Since $\terntr\leq f$ iff $\terntr^{-1}(1)\subseteq f^{-1}(1)$. and $f\leq \terncl$ iff $\terncl^{-1}(0)\subseteq f^{-1}(0)$, for every $\bit\in\{0,1\}$, we have
\[
\mbox{$\tern\in\{0,\unc,1\}^n$ is an $\bit$-witness of $f$ if and only if $f(\subcube)=\{\bit\}$.}
\]
Hence, a circuit $F$ computing a
Boolean function $f$ has a $\bit$-hazard iff $F(\tern)=\unc$ holds for
some $\bit$-witness $\tern\in\{0,\unc,1\}^n$ of $f$. It is
almost immediate that it is enough to only consider \emph{prime} witnesses.

\begin{prop}\label{prop1}
  Let $F$ be a DeMorgan circuit computing a Boolean function~$f$, and $\bit\in\{0,1\}$. If $F$ has a $\bit$-hazard, then $F$ has an $\bit$-hazard at some prime $\bit$-witness of~$f$.
\end{prop}

\begin{proof}
We only show the case $\bit=1$; the case $\bit=0$ is similar by considering the clause $\terncl$ instead of the term $\terntr$.
Suppose that the circuit $F$ has a $1$-hazard at
  some vector $\tern\in\{0,\unc,1\}^n$; hence, $f(\subcube)=\{1\}$ but
  $F(\tern)=\unc$. Since $\terntr(a)=1$ can only hold if $a\in
  \subcube$, we have $\terntr\leq f$, that is, $\terntr$ is an
  implicant of $f$. Then $\terntr$ must contain some \emph{prime}
  implicant $t$ of $f$ as a (not necessarily proper) subterm. This
  subterm is of the form $t=\tterntr$ for the vector
  $\ttern\in\{0,\unc,1\}^n$ obtained from $\tern$ by switching some
  stable bits to $\unc$. Since $\tterntr$ is a prime implicant of $f$,
  the vector $\ttern$ is a prime $1$-witness of $f$, and it remains to
  show that $f(\tsubcube)=\{1\}$ and $F(\ttern)=\unc$ hold. Since
  $\tterntr\leq f$, and since $\tterntr(a)=1$ can only hold if $a\in
  \tsubcube$, the equality $f(\tsubcube)=\{1\}$ follows. On the other
  hand, replacing stable bits $0/1$ by the unstable bit $\unc$ in the
  input vector $\tern$ cannot change the unstable output $\unc$ of
  DeMorgan circuit to a stable output $0$ or $1$ (see
  \cref{rem:mukaidono,rem:no-switch}).  So, $F(\tern)=\unc$ implies
  $F(\ttern)=\unc$, as desired.
\end{proof}

\subsection{Structure of $1$-hazards}
The following theorem gives us four necessary and sufficient  conditions
for a circuit to have a $1$-hazard. For a term $t$ and a clause $c$ we
write $t\cap\cl=\emptyset$ if $t$ and $\cl$ do not intersect, i.e., do
not share a literal in common.

\begin{thm}[$1$-hazards]\label{thm5}
  Let $F$ be a DeMorgan circuit computing a Boolean function~$f$, and
  $\tern\in\{0,\unc,1\}^n$ be a prime $1$-witness of $f$. The
  following assertions are equivalent.
  \begin{enumerate}
  \Item{1} $F$ has a $1$-hazard at $\tern$.

  \Item{2} $\cl(\tern)=\unc$ for some one-clause $\cl$ produced
    by~$F$.

  \Item{3} $\terntr\cap\cl=\emptyset$ for some one-clause $\cl$
    produced by~$F$.

  \Item{4} $t(\tern)\in\{0,\unc\}$ for every term $t$ produced by~$F$.

  \Item{5} The prime implicant $\terntr$ of $f$ is not produced by~$F$.
  \end{enumerate}
\end{thm}

\begin{proof}
Let $D$ be the formal DNF and $C$ the formal CNF of~$F$. Since $\tern$ is a $1$-witness of $f$, we have $f(\subcube)=\{1\}$.

  \equival{1}{2}: To show \implic{1}{2}, suppose that
  $F$ has a $1$-hazard at~$\tern$.  Hence, $C(\subcube)=\{1\}$ but
  $C(\tern)=\unc$. Since $C(\tern)=\unc$, $c(\tern)=\unc$ must hold
  for some clause $\cl$ of $C$, and it remains to show that this must be a one-clause. Suppose to the contradiction that $\cl(\tern)=\unc$ holds for some non-one clause $\cl$ of $C$.
   This can only
  happen if the vector $\tern$ evaluates every literal of $\cl$ to $0$
  or to $\unc$. Since $\cl$ has no variable together with its
  negation, negations of literals of $\cl$ evaluated to $\unc$ do not
  appear in $\cl$, and we can evaluate every literal of $\cl$ to
  $0$. On every such resolution $a\in\subcube$ of $\tern$, we will
  have $\cl(a)=0$ and, hence, also $C(a)=0$, a contradiction with
  $C(\subcube)=\{1\}$.

 To show the opposite implication \implic{2}{1}, suppose
  that $\cl_0(\tern)=\unc$ holds for some one-clause $\cl_0$ of $C$.
  Since $\tern$ is a $1$-witness of $f$, $C(\subcube)=\{1\}$ holds.
  So, since $1\land\unc=\unc$, it remains to show that
  $\cl(\tern)\in\{1,\unc\}$ holds for all clauses $\cl$ of $C$; then
  $F(\tern)=C(\tern)=\unc$, meaning that $F$ has a $1$-hazard at
  $\tern$.  Assume to the contrary that $\cl(\tern)=0$ holds for some
  clause $\cl$ of $C$. Since $0\lor\unc=\unc\neq 0$, this means that
  the vector $\tern$ evaluates to $0$ all literals of $\cl$.  But then
  every resolution of $a\in\subcube$ of $\tern$ also evaluates these
  literals to $0$, and we obtain $c(a)=0$, a contradiction with
  $C(a)=1$.

  \equival{2}{3}: Let $\cl$ be a one-clause of $C$.
  If $\cl(\tern)=\unc$, then $z(\tern)\in\{0,\unc\}$ holds for all
  literals $z$ of $c$. But $\terntr(\tern)=1$ implies that the vector $\tern$ evaluates all literals of
  $\terntr$ to $1\not\in\{0,\unc\}$; hence,
  $\terntr\cap\cl=\emptyset$. To show the opposite implication \implic{3}{2}, note that (by the definition of the term $\terntr$)
  the vector
  $\tern$ evaluates to $\unc$ all variables not in $\terntr$.
  So, if $\terntr\cap\cl=\emptyset$, then $\tern$ evaluates every literal of $\cl$ to $0$ of to $\unc$. Since $\cl$ is a one-clause, it contains some variable $x_i$ together with its negation $\bar{x}_i$. Since $\cl(\tern)\neq 1$, we have $x_i(\alpha)=\unc$ and, hence, also $\cl(\tern)=\unc$.

  \equival{1}{4}: If $F$ has a $1$-hazard at $\tern$, then
  $D(\tern)=F(\tern)=\unc$ holds. So, since $1\lor\unc=1\neq \unc$,
  $t(\tern)\in\{0,\unc\}$ must hold for all terms $t$ of $D$. To show
  the opposite implication \implic{4}{1}, suppose that
  $t(\tern)\in\{0,\unc\}$ holds for every term $t$ of $D$. Our goal is
  to show that then the circuit $F$ has a $1$-hazard at~$\tern$.
  Since $f(\subcube)=\{1\}$ ($\tern$ is an implicant of $f$), it is
  enough to show that $t(\tern)=\unc$ holds for at least one term
  $t$ of~$D$. Suppose to the contrary that $t(\tern)=0$ holds for all
  terms $t$ of $D$. Then the vector $\tern$ evaluates to $0$ at least
  one literal in every term of $D$. But then also every resolution
  $a\in\subcube$ of $\tern$ evaluates to $0$ at least one literal in
  every term of $D$, and we have $D(\subcube)=\{0\}\neq\{1\}$, a
  contradiction with $f(\subcube)=\{1\}$.

  \equival{4}{5}: If the prime implicant $\terntr$ of $f$ is a term
  of~$D$, then $\terntr(\tern)=1\not\in\{0,\unc\}$ holds for this
  term. To show the opposite implication \implic{5}{4}, suppose that
  $t(\tern)=1$ holds for some term $t$ of $D$. Then
  $t(\subcube)=\{1\}$, that is, $\subcube\subseteq t^{-1}(1)$. Since
  $\subcube=\terntr^{-1}(1)$, this yields the inclusion
  $\terntr^{-1}(1)\subseteq t^{-1}(1)$, which can only hold if $t$ is
  a subterm of $\terntr$.  Since $t$ is an implicant of $f$ and
  $\terntr$ is a \emph{prime} implicant of $f$, this is only possible
  if $t=\terntr$. Hence, the prime implicant $\terntr$ is a term
  of~$D$, as desired.
\end{proof}

\begin{ex}\label{ex:thm5}
  Consider the circuit $F=x(\bar{y}\lor \bar{z})\lor \bar{x}y$
  computing the Boolean function $f(x,y,z)=x\bar{y}\lor
  x\bar{z}\lor \bar{x}y$. The formal DNF of $F$ is
  $D=x\bar{y}\lor x\bar{z}\lor \bar{x}y$ and the formal CNF is
  $C =(x\lor\bar{x})(\bar{x}\lor\bar{y}\lor\bar{z})(x\lor y)(y\lor \bar{y}\lor\bar{z})$.
   According to \cref{thm5}, the circuit $F$ has
  a $1$-hazard by either of the following four reasons, where
  $\tern=(\unc,1,0)$ is a prime $1$-witness of $f$, and
  $c=x\lor\bar{x}$ is a one-clause produced by $F$:
  \begin{itemize}
  \item[$\circ$] $c(\tern)=\unc\lor\bar{\unc}=\unc$;
  \item[$\circ$] $y\bar{z}\cap (x\lor\bar{x})=\emptyset$;
  \item[$\circ$] $x\bar{y}(\tern)=0\neq 1$ and $x\bar{z}(\tern)=\unc\neq 1$ and
    $\bar{x}y(\tern)=\bar{\unc}=\unc\neq 1$;
  \item[$\circ$] the prime implicant $y\bar{z}$ of~$f$ is not produced
    by~$F$.
  \end{itemize}
  And indeed, on the vector $\tern=(\unc,1,0)$, we have
  $F(\tern)=\unc(0\lor 1)\lor \bar{\unc}=\unc\lor\bar{\unc}=\unc$,
  even though $F(0,1,0)=F(1,1,0)=1$.
\end{ex}

\subsection{Structure of $0$-hazards}
\label{sec:dual}
There is also an analogue of \cref{thm5} for $0$-hazards (\cref{thm6} below), and it can be proved using ``dual'' arguments: just
interchange the roles of constants $0$ and $1$ as well as of terms and
clauses in the proof of \cref{thm5}. But to stress the duality between $0$-hazards and $1$-hazards
as well as between formal DNFs and formal CNFs, we will show that \cref{thm6} \emph{itself} is the dual version of \cref{thm5}.

Recall that the \emph{dual} of a Boolean function $f(x_1,\ldots,x_n)$
is the
Boolean function $\dual{f}(x):=\neg f(\bar{x})$, where $\bar{x}=
(\bar{x}_1,\ldots,\bar{x}_n)$. That is, we negate each input bit as
well as the result.  By the DeMorgan laws $\neg(x\lor
y)=\bar{x}\land\bar{y}$ and $\neg(x\land y)=\bar{x}\lor\bar{y}$, we
have $\dual{(f\lor g)}=\dual{f}\land\dual{g}$, $\dual{(f\land
  g)}=\dual{f}\lor\dual{g}$ and $\dual{(\neg f)}=\neg\dual{f}$. It is
well known and easy to show (see, e.g., \cite[Theorem~4.6]{hammer})
that (prime) implicates of a Boolean function $f(x)$ are (prime)
implicants of the dual function $\dual{f}(x)=\neg f(\bar{x})$, and
vice versa.

The \emph{dual} $\dual{F}$ of a DeMorgan circuit $F$ is obtained by
exchanging the gates AND and OR, as well as the input constants $0$
and $1$. For example, the dual of the circuit
$F=(\bar{x}\lor y)(y\lor\bar{z})\lor xyz$ is the circuit
$\dual{F}=(\bar{x}y\lor y\bar{z})(x\lor y\lor z)$.  In particular, the
dual of a clause $\cl=z_1\lor\cdots\lor z_m$ is the term
$\dual{\cl}=z_1\land\cdots\land z_m$, and vice versa.
Also, the dual of a
DNF $D=t_1\lor\cdots\lor t_l$ is the CNF $C=\dual{t}_1\land\cdots\land
\dual{t}_l$.  That is, in the ``dual world,'' the roles of constants
$0$ and $1$ as well as of operations AND and OR are interchanged.
Using  DeMorgan laws, it
is easy to show that a circuit $F$ computes a Boolean function $f$ iff
the dual circuit $\dual{F}$ computes $\dual{f}$ (see, for example,
\cite[Theorem~1.3]{hammer}).  Since the DeMorgan laws hold over the
ternary domain $\{0,\unc,1\}$, we have $\dual{F}(\tern)=\neg
F(\bar{\tern})$ for every ternary vector $\tern\in\{0,\unc,1\}^n$.
Finally, note that $(\dual{f})^d=f$ and  $(\dual{F})^d=F$.

\begin{ex}
The dual  $\dual{F}=x\bar{z}\lor yz$ of the circuit $F=(x\lor\bar{z})(y\lor z)$ computing the multiplexer function $f=xz\lor y\bar{z}$ computes $\dual{f}=x\bar{z}\lor yz$. The circuit $F$ produces the clause $\cl=x\lor\bar{z}$ while the circuit
$\dual{F}$ produces the term $\dual{\cl}=x\bar{z}$. The circuit $F$ has a $0$-hazard at $\tern=(0,0,\unc)$ while the circuit
$\dual{F}$ has a $1$-hazard at $\ntern=(1,1,\unc)$.
\end{ex}

It is not difficult to show that such a duality between the produced clauses and terms as well as between $0$-hazards and $1$-hazards
also holds in general.

\begin{lem}\label{lem:dual}
Let $F$ be a DeMorgan circuit, and $\tern\in\{0,\unc,1\}^n$.
\begin{enumerate}
\Item{i} A clause $\cl$ is produced by $F$ if and only if the term
    $\dual{\cl}$ is produced by~$\dual{F}$.
\Item{ii} $F$ has a $0$-hazard at $\tern$ if and only if $\dual{F}$
    has a $1$-hazard at~$\bar{\tern}$.
\end{enumerate}
\end{lem}

\begin{proof}
(i): Easy induction on the size of the circuit $F$. The induction
  basis (when the circuit is a single literal or a constant) is
  trivial. For the induction step, suppose that the claim holds for
  all circuits of size at most $s-1$, and let $F$ be a circuit of
  size~$s$. Suppose that a clause $\cl$ is produced by~$F$.

  If $F=F_1\land F_2$, then $\cl$ is produced by the circuit $F_i$ for
  some $i\in\{1,2\}$. By the induction hypothesis, the term
  $\dual{\cl}$ is produced by $\dual{F}_i$. The dual of the circuit
  $F$ in this case is $\dual{F}=\dual{F_1}\lor\dual{F_2}$. So, the
  term $\dual{\cl}$ is produced by~$F$ as well.

  If $F=F_1\lor F_2$, then $\cl=\cl_1\lor \cl_2$ for some clauses
  $\cl_1$ and $\cl_2$ produced by $F_1$ and~$F_2$. By the induction
  hypothesis, the terms $t_1=\dual{\cl}_1$ and $t_2=\dual{\cl}_2$ are
  produced by the dual circuits $\dual{F}_1$ and $\dual{F}_2$. Hence,
  the term $t=t_1\land
  t_2=\dual{\cl}_1\land\dual{\cl}_2=\dual{(\cl_1\lor
    \cl_2)}=\dual{\cl}$ is produced by the dual circuit
  $\dual{F}=\dual{F_1}\land\dual{F_2}$.  This shows the
  $(\Rightarrow)$ direction of claim~(i); the opposite $(\Leftarrow)$
  direction can shown be via the same argument by interchanging ANDs
  and ORs.

  (ii): A Boolean vector $a$ is a resolution of $\tern$ iff its complement $\bar{a}$ is a resolution of $\ntern$. That is, $a\in\subcube$ iff $\bar{a}\in\nsubcube$. Since $F(a)=0$ iff $\dual{F}(\bar{a})=\neg F(a)=1$, we have $F(\subcube)=\{0\}$ if and only if $\dual{F}(\nsubcube)=\neg F(\subcube)=\{1\}$.
  On the other hand, $F(\tern)=\unc$ holds iff $t(\tern)\in\{0,\unc\}$ holds for all terms and $t(\tern)=\unc$ holds for at least one term $t=z_1\land\cdots \land z_m$ produced by $F$. By claim (i), the clauses produced by $\dual{F}$ are the duals $\dual{t}=z_1\lor\cdots \lor z_m$ of terms $t$ produced by $F$. Since $t(\tern)\in\{0,\unc\}$ iff $\dual{t}(\ntern)\in\{1,\unc\}$, and since $t(\tern)=\unc$ iff $\dual{t}(\ntern)=\unc$, we have $F(\tern)=\unc$ iff $\dual{F}(\ntern)=\unc$. Thus, $F$ has a $0$-hazard at $\tern$ if and only if $\dual{F}$  has a $1$-hazard at~$\bar{\tern}$.
\end{proof}

\begin{thm}[$0$-hazards]\label{thm6}
  Let $F$ be a DeMorgan circuit computing a Boolean function~$f$, and
  $\tern\in\{0,\unc,1\}^n$ be a prime $0$-witness of $f$. The
  following assertions are equivalent.
  \begin{enumerate}
  \Item{1} $F$ has a $0$-hazard at $\tern$.

  \Item{2} $t(\tern)=\unc$ for some zero-term $t$ produced by~$F$.

  \Item{3} $\terncl\cap t=\emptyset$ for some zero-term $t$
    produced by~$F$.

  \Item{4} $\cl(\tern)\in\{1,\unc\}$ for every clause $\cl$ produced
    by~$F$.

  \Item{5} The prime implicate $\terncl$ of $f$ is not produced by~$F$.
  \end{enumerate}
\end{thm}

\begin{proof}
For notational simplicity, assume w.l.o.g. that $\tern$ is of the form $\tern=(a_1,\ldots,a_m,\unc,\ldots,\unc)$ with all $a_i\in\{0,1\}$.
Then the clause associated with the vector $\tern$ is $\terncl=x_1^{\bar{a}_1}\lor \cdots \lor x_m^{\bar{a}_m}$, and the term associated with complementary vector $\ntern=(\bar{a}_1,\ldots,\bar{a}_m,\unc,\ldots,\unc)$ is the dual
$\nterntr=x_1^{\bar{a}_1}\land \cdots \land x_m^{\bar{a}_m}=\dual{\terncl}$ of the clause $\terncl$. It is well known and easy to show
  (see, e.g., \cite[Theorem~4.1]{hammer}) that for any two Boolean
  functions $g$ and $f$, we have $f\leq g$ if and only if
  $\dual{g}\leq\dual{f}$. Thus, $f\leq \terncl$ holds iff
  $\nterntr\leq \dual{f}$ holds, that is, the vector $\tern$ is a
  prime $0$-witness of~$f$ if and only if  the vector $\ntern$ is a
  prime $1$-witness of~$\dual{f}$.

When applied to the dual circuit
$\dual{F}$ computing the dual function $\dual{f}$, \cref{thm5} implies
that following assertions are equivalent.

\begin{enumerate}
  \Item{1*} The circuit $\dual{F}$ has a $1$-hazard at $\ntern$.

  \Item{2*} $\cl(\ntern)=\unc$ for some one-clause $\cl$ produced
  by~$\dual{F}$.

  \Item{3*} $\nterntr\cap\cl=\emptyset$  for some one-clause
  $\cl$ produced by~$\dual{F}$.

  \Item{4*} $t(\ntern)\in\{0,\unc\}$ for every term $t$ produced
  by~$\dual{F}$.

  \Item{5*} The prime implicant $\nterntr$ of $\dual{f}$ is not
  produced by~$\dual{F}$.
\end{enumerate}
It is therefore enough to show that the corresponding assertions are
equivalent. The equivalences \equival{1}{1*} and \equival{5}{5*}
follow directly from \cref{lem:dual}. The remaining three equivalences
also follow from \cref{lem:dual}(i) and the following simple observations.
 We have $\cl(\tern)=0$ for a clause $\cl$ iff the vector $\tern$ evaluates all literals of $\cl$ to $0$, which happens precisely when
 the complementary vector $\ntern$ evaluates all these literals to $1$, that is, when $t(\ntern)=1$ holds for the term $t=\dual{c}$.
So, $\cl(\tern)\in\{1,\unc\}$ iff $t(\ntern)\in\{0,\unc\}$, and
 $\cl(\tern)=\unc$ iff $t(\ntern)=\unc$. Since a clause is a one-clause iff its dual is a zero term, the equivalences \equival{2}{2*} and \equival{4}{4*} follow.
Finally, the dual $t=\dual{\cl}$ of any
one-clause $\cl$ from claim (3*) is a zero-term produced by $F$.
Note that both $t$ and $\cl$ have the same literals. The sets of literals of the term $\nterntr=x_1^{\bar{a}_1}\land \cdots \land x_m^{\bar{a}_m}$ and of the clause $\terncl=x_1^{\bar{a}_1}\lor \cdots \lor x_m^{\bar{a}_m}$ are also the same.
So,  $\nterntr\cap\cl=\emptyset$ holds precisely when $\terncl\cap t=\emptyset$ holds, and the equivalence \equival{3}{3*} follows as well.
\end{proof}

As we already mentioned in the introduction, a classical result of  Huffman~\cite{huffman} is that  the DNF whose terms are prime implicants of a Boolean function $f$ is a hazard-free circuit computing~$f$. An also classical result of Eichelberger~\cite[Theorem~2]{Eichelberger} extends Huffman's theorem.

\begin{nthm}
  Let $F$ be a DeMorgan circuit computing a Boolean function~$f$.  If
  $F$ produces no zero-terms, then $F$ is hazard-free if and only if
  the circuit~$F$ produces all prime implicants of~$f$.
\end{nthm}
This theorem is stated and proved in~\cite{Eichelberger} only for zero-term free DNFs (depth-two circuits), but it also holds for DeMorgan
circuits producing no zero-terms: formal DNFs of such circuits do not
have such terms. Using \cref{lem:dual,thm5,thm6}, we can remove the ``no zero-terms'' restriction from Eichelberger's theorem.

\begin{cor}[Extended Eichelberger's Theorem]\label{cor:ext-Eichenberger}
Let $F$ be a DeMorgan circuit computing a Boolean function $f$. The following assertions are equivalent.
\begin{enumerate}
\Item{1} $F$ is hazard-free.
\Item{2} Neither $F$ nor $\dual{F}$ has a $0$-hazard.
\Item{3} Neither $F$ nor $\dual{F}$ has a $1$-hazard.
\Item{4} $F$ produces all prime implicants and all prime implicates of~$f$.
\end{enumerate}
\end{cor}

\begin{proof}
The equivalences \equival{1}{2} and  \equival{1}{3} are  direct consequences of \cref{lem:dual}. The implication \implic{1}{4} follows directly from \cref{thm5,thm6}. To show the opposite implication \implic{4}{1}, suppose that $F$ produces all prime implicants and all prime implicates of~$f$. Then, by \cref{thm5,thm6}, the circuit $F$  has no hazards at prime witnesses $\tern\in\{0,\unc,1\}^n$ of $f$ and, by \cref{prop1}, has no hazards at any ternary vectors either.
\end{proof}

\section{Final remarks}
\label{sec:final}
Recall that $\uL{f}$ denotes the minimum number of gates in a DeMorgan
circuit computing a Boolean function~$f$. We also have introduced
versions of this complexity measure: $\hL{f}$ (the circuit must be
hazard-free), $\zpL{f}$ (the circuit must produce all prime implicants
of $f$ but must produce no zero-terms), and $\mL{f}$ (the circuit must
be monotone, i.e. must have no negated variables as inputs). Finally, let $\hL{n}$ be the
Shannon function for hazard-free circuits, that is, the maximum of
$\hL{f}$ over all Boolean functions $f$ of $n$ variables.

By \cref{cor:lowerA}, we know that for every Boolean function~$f$ the
inequalities $\mL{\up{f}}\leq \hL{f}\leq \zpL{f}$ hold. So, the following
natural questions arise.

\begin{enumerate}
\item How large can the gap $\hL{f}/\mL{\up{f}}$ be?
\item How large can the gap $\zpL{f}/\hL{f}$ be?
\item What is the asymptotic value of $\hL{n}$? Is it $\hL{n}\sim
  2^n/n$?
\end{enumerate}

Let us show that we already know the \emph{order} of growth of the function $\hL{n}$: it is $2^n/n$. This can be shown by extending the construction of Shannon~\cite{shannon49} and
Muller~\cite{muller} to show that $\uL{n}=O(2^n/n)$ holds for
unrestricted circuits. They used the recursion
$F(x_1,\ldots,x_n)=\bar{x}_n\cdot F_0\lor x_n\cdot F_1$, where
 $F_0$ and $F_1$ are DeMorgan circuits computing the subfunctions
$f_0=f(x_1,\ldots,x_{n-1},0)$ and $f_1=f(x_1,\ldots,x_{n-1},1)$ of a given Boolean function $f(x_1,\ldots,x_n)$.

The circuit $F$ constructed using this
decomposition will produce no zero-terms and, by \cref{prop:zero-terms},  will have no $0$-hazards. But it can have $1$-hazards even if  both subcircuits $F_0$ and $F_1$ are hazard-free.
 For this to happen, it is enough that some prime
implicant $t$ of $f$ contains neither $x_n$ nor $\bar{x}_n$; then $t$ is not produced by $F$ and, by \cref{thm5}, $F$ has a $1$-hazard. And indeed, we
can take any ternary vector $\tern=(a,\unc)$ with $a\in\{0,1\}^{n-1}$
and $t(a)=1$. Since then also $F_0(a)=1$ and $F_1(a)=1$, we have
$F(\subcube)=\{1\}$.  But $F(\tern)=\bar{\unc}\cdot 1\lor \unc\cdot
1=\unc$.

A classical idea to generate prime implicants, rediscovered by several independent researchers---Blacke~\cite{Blake}, Samson and Mills~\cite{Samson}, Quine~\cite{Quine}---is to use the \emph{consensus} recursion $F=\bar{x}_n\cdot F_0\lor x_n\cdot F_1\lor F_0\cdot F_1$: if a term $t$ is an implicant of $f$, and if $t$ contains neither $\bar{x}_n$ nor $x_n$, then $t$ is an implicant of both subfunctions $f_0$ and $f_1$ of~$f$.

It is, however, not a priori clear that this extended recursion will not introduce $0$-hazards: even if neither of the subcircuits $F_0$ and $F_1$ produces
zero-terms, the additional subcircuit $F_0\cdot F_1$ could, in general, produce zero-terms and
such terms may lead to $0$-hazards (see \cref{thm5}). Fortunately, these zero-terms are ``innocent,'' as long as both   circuits
$F_0$ and $F_1$ are hazard-free.

\begin{prop}\label{prop:decomp}
Let $F_0(x_1,\ldots,x_{n-1})$ and $F_1(x_1,\ldots,x_{n-1})$ be arbitrary hazard-free circuits, and $x_n$ be a new variable. Then the circuit
\begin{equation}\label{eq:rec}
  F(x_1,\ldots,x_n)=\bar{x}_n\cdot
  F_0\lor x_n\cdot F_1\lor F_0\cdot F_1
\end{equation}
 is hazard-free.
\end{prop}

\begin{proof}
Assume that the circuit $F$ has a hazard at
some vector $\tern\in\{0,\unc,1\}^n$; hence, $F(\subcube)=\{\bit\}$
for some $\bit\in\{0,1\}$ but $F(\tern)=\unc$.  Our goal is to show
that then at least one of the circuits $F_0$ and $F_1$ must have an
$\bit$-hazard at~$\tern$.

\Case{1} $\tern_n\in\{0,1\}$, say, $\tern_n=0$. Then for every resolution $a\in\subcube$ of~$\tern$, we have $F(a)=F_0(a)\lor F_0(a)\cdot F_1(a)=F_0(a)$; hence, $F_0(\subcube)=F(\subcube)=\{\bit\}$.
Since $\tern_n=0$, we have $x_n\cdot F_1(\tern)=0$ and, since the absorbtion law $x\lor xy=x$ holds also over $\{0,\unc,1\}$, we obtain $F(\tern)=F_0(\tern)\lor F_0(\tern)\cdot F_1(\tern)=F_0(\tern)$.
Thus, $F_0(\tern)=F(\tern)=\unc$, meaning that the circuit $F_0$ has a $\bit$-hazard at~$\tern$.

\Case{2} $\tern_n=\unc$. Since the circuits $F_0$ and $F_1$ do not depend on $x_n$, both $F_0(\subcube)=\{\bit\}$ and $F_1(\subcube)=\{\bit\}$  must hold in this case. Indeed, if say, $F_0(a)=\bar{\bit}$ for some resolution $a\in\subcube$ of $\tern$, then also $F_0(a')=\bar{\bit}$ for the resolution $a'=(a_1,\ldots,a_{n-1},0)$ of $\tern$, and we obtain
$F(a')=1\cdot \bar{\bit}\lor \bar{\bit}\cdot F_1(a')=\bar{\bit}\neq\bit$. Thus, both $F_0(\subcube)=\{\bit\}$ and $F_1(\subcube)=\{\bit\}$ hold. Then both values $F_0(\tern)$ and $F_1(\tern)$ must belong to $\{\bit,\unc\}$ (see~\cref{rem:no-switch}). Since $\unc\land 0=0$ and $\unc\lor 1=1$, both
$F_0(\tern)=\bit$ and $F_1(\tern)=\bit$ cannot hold because then
$F(\tern)=\bar{\unc}\cdot\bit\lor \unc\cdot\bit\lor \bit=\bit\neq\unc$. So, $F_i(\tern)=\unc$ holds for some
$i\in\{0,1\}$, meaning that the circuit $F_i$ has a $\bit$-hazard at
$\tern$.
\end{proof}

When directly applied, the recursion \cref{eq:rec} yields $\hL{n}= O(2^n)$. Nitin Saurabh (personal communication) suggested to combine
the hazard-freeness preserving recursion
\cref{eq:rec} with the argument used by Shannon~\cite{shannon49} and
Muller~\cite{muller} to show $\uL{n}=O(2^n/n)$ for unrestricted circuits. And indeed, the combination yields much better upper bound
$\hL{n}= O(2^n/n)$.

Take an arbitrary Boolean function $f(x_1,\dots,x_n)$, and apply
the recursion \cref{eq:rec} for $n-m$ steps to obtain a hazard-free
circuit $F_{n-m}$ of size $5\cdot 2^{n-m}$ computing the function $f$ from all
its $2^{n-m}$ subfunctions
$f_b(x_1,\ldots,x_m)=f(x_1,\ldots,x_m,b_1,\ldots,b_{n-m})$ on the
first $m$ variables; here, $m\leq n$ is a parameter to be specified latter.
Inputs to $F$ are Boolean functions from the set
$\HH_m$ of all $|\HH_m|=2^{2^m}$ Boolean functions $h(x_1,\ldots,x_m)$
on the first $m$ variables. Shannon's idea is this: if $2^{n-m}\gg
2^{2^m}$, then same functions from $\HH_m$ will appear
many times among the inputs of $F$. It is then more economical to
simultaneously compute all the functions in $\HH_m$ once beforehand,
rather than to recompute the residual functions $f_b$ at each of the
$2^{n-m}$ inputs of $F$.

\circuits

Using the recursion \cref{eq:rec}, we can construct a hazard-free
circuit $\U_m$ of size at most $5\cdot 2^{2^m}$ which simultaneously computes all
$2^{2^m}$ Boolean functions in~$\HH_m$: given the circuit $\U_{m-1}$,
we can use the recursion \cref{eq:rec} to obtain the circuit $\U_m$ by
adding five gates per one function in~$\HH_m$ (see \cref{fig:shannon}, left). So, the size of the
resulting circuit $\U_m$ is at most $5$ times
$\sum_{i=2}^{m-1}2^{2^i}\leq 2^{2^m}$.

By identifying the input gates
$f_b\in\HH_m$ of the circuit $F_{n-m}$ with the corresponding output gates
of $\U_m$, we obtain a hazard-free circuit $F$ for $f$.
That is, the circuit $F$ is obtained from the circuit $F_{n-m}$ by further applying the hazard-freeness preserving recursion \cref{eq:rec} (see \cref{fig:shannon}, right): we only do not repeat the construction for the same input subfunctions~$f_b$.
By \cref{prop:decomp}, the obtained circuit $F$ is hazard-free.
The number of gates in $F$ is at most five times $2^{2^m}+2^{n-m}$. For $m=\log_2 (n-\log_2 n)$, we have $2^{2^m}=2^n/n$ and
$2^{n-m}=2^n/(n-\log_2n)=c\cdot 2^n/n$, where $c=\left(1+\tfrac{\log_2n}{n-\log_2n}\right)$. So, $\hL{n}=O(2^n/n)$ follows by takings $m$ to be a nearest to $\log_2 (n-\log_2 n)$ integer.  Since $\hL{n}\geq \uL{n}=\Omega(2^n/n)$, the order
of magnitude of the Shannon function for hazard-free circuits is
already known: $\hL{n}=\Theta(2^n/n)$.  It remains, however, open
whether $\hL{n}\sim 2^n/n$ holds.

\section*{Acknowledgments}
I am thankful to Nitin Saurabh and Igor Sergeev for enlightening
discussions.


\end{document}